\documentclass[12pt,a4paper]{article}
\usepackage{graphicx}
\usepackage[]{amsmath,amssymb,amsthm,latexsym}
\usepackage{hyperref}
\usepackage[T1]{fontenc}

\newcommand{\pr}{\mathrm{pr}}
\newcommand{\prfst}{\mathrm{pr,FST}}
\newcommand{\be}{\begin{enumerate}}
\newcommand{\ee}{\end{enumerate}}
\newcommand{\bi}{\begin{itemize}}
\newcommand{\ei}{\end{itemize}}

\newtheorem{theorem}{Theorem}[section]
\newtheorem{proposition}[theorem]{Proposition}
\newtheorem{lemma}[theorem]{Lemma}
\newtheorem{corollary}[theorem]{Corollary}
\theoremstyle{definition}
\newtheorem{definition}[theorem]{Definition}
\theoremstyle{remark}
\newtheorem{remark}[theorem]{Remark}
\newcommand{\bdf}{\begin{definition}}
\newcommand{\edf}{\end{definition}}
\newcommand{\bprop}{\begin{proposition}}
\newcommand{\eprop}{\end{proposition}}
\newcommand{\btheor}{\begin{theorem}}
\newcommand{\etheor}{\end{theorem}}
\newcommand{\bcorol}{\begin{corollary}}
\newcommand{\ecorol}{\end{corollary}}
\newcommand{\blemma}{\begin{lemma}}
\newcommand{\elemma}{\end{lemma}}

\newcommand{\R}{\mathbb{R}}

\newcommand{\norm}[1]{\left\lVert#1\right\rVert}

\usepackage{tikz}    % for LaTeX-drawn pictures.
\usetikzlibrary{decorations.pathreplacing,decorations.markings,positioning,shapes.misc}
\usepackage{pgfplots}
\pgfplotsset{compat=1.17} % to ensure pgf stability against future versions

\usepackage{subfigure}

\begin{document}

\begin{center}
{\large{\bf 
Walrasian equilibria are almost always finite in number}}\\
\mbox{} \\
\begin{tabular}{cc}
{\bf Sofia B.\ S.\ D.\ Castro$^{\diamondsuit, *}$} & {\bf Peter B.\ Gothen$^{\spadesuit}$} \\
{\small sdcastro@fep.up.pt} & {\small pbgothen@fc.up.pt}  \\
{\small orcid: 0000-0001-9029-6893} & {\small orcid: 0000-0002-4624-3871}
\end{tabular}

\end{center}

\noindent {\small $\diamondsuit$} CMUP, Cef.UP and Faculdade de Economia do Porto, Universidade do Porto, Rua Dr. Roberto Frias, 4200-464 Porto, Portugal.

\noindent {\small *} corresponding author

\bigbreak

\noindent {\small $\spadesuit$} CMUP and Departamento de Matem\'atica, Faculdade de Ci\^encias, Universidade do Porto, Rua do Campo Alegre, 4169-007 Porto, Portugal.

\bigbreak

\begin{abstract}
In the context of exchange economies, defined by aggregate excess demand functions, we extend results on finiteness of equilibria to economies defined on the full open price simplex. Genericity is proved also for critical economies and, in both cases, in the strong sense that it holds for an open dense subset of economies in the Whitney topology.  We use the concept of finite singularity type from singularity theory. This concept ensures that the number of equilibria of a map appear only in finite number. We then show that maps of finite singularity type make up an open and dense subset of all smooth proper maps. We translate the result to the set of aggregate excess demand functions of an exchange economy to show that finiteness of equilibria is a generic property in sets of economies which form an open subset of the space of proper maps. We construct an explicit class of aggregate excess demand for such economies spanned by Cobb-Douglas consumers.

Along the way, we explore the extension of the classical results of Sonnenschein-Mantel-Debreu to functions defined on the full open price simplex, rather than just compact subsets of the simplex. In particular, we identify necessary boundary conditions for such functions to be aggregate excess demand functions. 
\end{abstract}

\noindent {\em Keywords:} Sonnenschein-Mantel-Debreu, general equilibrium, aggregate excess demand, determinacy, genericity

\vspace{.3cm}

\noindent {\em JEL classification:} C02, D50, C62

\vspace{.5cm}

\section{Introduction}

We find conditions under which the set of economies with a finite number of Walrasian equilibria are generic in the sense that they form an open and dense set.
Our results apply to both regular and critical economies. This means that economies with an infinite number of equilibria are difficult to encounter.
Our results are an informative contribution to the fundamental question of the number of equilibria in an exchange economy initiated in the 1970's.
Although our results do not change the current perception of economic equilibrium, they do transform this perception into a solid result thus providing a firm foundation for this relevant aspect of the economic science.

Our results add to those in Debreu \cite{Debreu70}, Allen \cite{Allen84}, Mas-Colell and Nachbar \cite{M-CN}, and more recently Castro {\em et al.} \cite{CDG2012}. In this latter work the authors prove that the set of all economies defined by an aggregate excess demand (AED) on the trimmed simplex (that is, such that every price remains bounded away from zero) and with a finite number of equilibria is residual (that is, a countable intersection of open and dense sets) in the space of all maps for the compact-open topology. In other words, having a finite number of equilibria is a generic property for an economy defined on the trimmed simplex.
Our results also complement satisfactorily the existence results of Bhowmik and Yannelis \cite{BhowmikYannelis2024}.

Several technical points prevented the authors of \cite{CDG2012} from extending the result to economies defined in the full open price simplex. 
These are partially overcome in the current article where the results hold for the full open simplex. 
The main difficulty lies in the fact that the fundamental result of Sonnenschein \cite{Sonnenschein1972}, Mantel \cite{Mantel1974}, and Debreu \cite{Debreu74} does not hold for functions defined on the open simplex. Indeed, Balasko \cite{Balasko86} shows that an AED must be a proper map and that its Brouwer degree is $\pm 1$ (the sign depending on the dimension of the commodity space). We identify additional specific necessary boundary conditions on AEDs as prices approach the boundary of the simplex.

Our results hold for the Whitney topology, rather than the compact-open topology.
We note that, in the present context of the open simplex, the Whitney topology is the more appropriate topology because, by its very definition, the compact-open topology is not sensitive to the behaviour of functions close to the boundary of the simplex.
Besides, in the Whitney topology, we obtain a stronger genericity result on finiteness of equilibria, showing that it holds for an open and dense space of economies, rather than the weaker notion of residual. To illustrate why open and dense is a stronger notion of generic, recall that in the simpler context of measure spaces the complement to such a subset has measure zero, whereas this is not necessarily the case for the complement of a residual set.

Our approach rests on identifying the space of economies, via their aggregate excess demand functions, with a subspace of the space of proper maps defined on the open price simplex. 
Indeed, we show that finiteness of equilibria holds for an open and dense subset of any open set of economies in the space of proper maps.

It is natural to ask whether the space of all AEDs forms a Whitney open set of the space of all proper maps. In view of the constraints already mentioned on AEDs, an affirmative answer could be viewed as the natural extension of the Sonnenschein-Mantel-Debreu result to functions defined on the open simplex. Together with the results of the present paper, namely those in Section~\ref{sec:finiteness}, such an answer would imply genericity of finiteness of equilibria in full generality for economies defined on the open price simplex.
While we do not answer this question in full generality, we do construct a class of such economies stemming from a finite number of Cobb-Douglas utility maximising consumers in a way to be made precise in Theorem~\ref{thm:CD-spans}.
We conjecture that this construction can be extended to other classes of utility functions.
Related examples using Cobb-Douglas consumers appear in Aloqeili~\cite{Aloqeili2005} and Chiappori and Ekeland~\cite{ChiapporiEkeland2004}.

It is important to stress that our results rely on the use of singularity theory, an area of mathematics that has not been commonly used in economics. Moreover, we use singularity theory results that first appeared in the 1990's\footnote{We became aware of the relevant results in singularity theory, namely the concept of finite singularity type, by chance while having a conversation with Andrew du Plessis, the first author of \cite{dPW}.}, well after economists had settled for an incomplete answer to the question of finiteness of equilibria.

The next section presents the technical concepts we use. These repeat some parts of \cite{CDG2012} but we include them here for ease of reference. 
We also state and prove the necessary boundary conditions for AEDs as prices tend to the boundary of the price simplex.
Section~\ref{sec:SMD-open-simplex} identifies classes of AEDs for which genericity of finiteness of equilibria holds.
It also presents a construction of such a class based on economies with Cobb-Douglas consumers.
Section~\ref{sec:finiteness} contains our main results and the final section concludes.

\section{Preliminaries and boundary conditions}
\label{sec:preliminary}

We start by defining notation to be used throughout this article.

Let $\R_+$ denote the set of non-negative numbers and let $\R_{++}$ denote the set of strictly positive numbers.

Let $\Delta = \{ p \in \R^{\ell}_+: \; \sum_{i=1}^{\ell} p_i = 1\}$
be the unit simplex in $\R^{\ell}$. This is the full price simplex for
the economy.
Its interior is denoted $\mathring{\Delta}$ and   referred to as the open price simplex.
Let $S^{\ell-1}_{++} = \{p\in \R_{++}^\ell: \norm{p}=1\}$ denote the
intersection of the unit sphere with the strictly positive orthant. This is the space of normalised positive prices.
By a standard renormalisation procedure, we can use the sphere $S^{\ell-1}_{++}$ and the open simplex $\mathring{\Delta}$ interchangeably.

These sets are used to define an economy with a finite number of consumers and goods. 
The consumers are represented by smooth utility functions. We assume that the underlying preference relations satisfy the boundary condition of Mas-Colell \cite[Def.~2.3.16]{Mas-Colell85}, also used in \cite{CDG2010}. 
Thus indifference sets are complete submanifolds of $\R^{\ell}_{++}$. Our hypotheses exclude corner solutions and guarantee that the utility maximisation problem has a solution for every $p$.
We also assume that there is some strictly positive amount of each good in the economy.

From Proposition 3.3 in \cite{CDG2010}, there is a homeomorphism between the space of utilities and the space of preferences describing the economy.

Such an economy has an aggregate excess demand function $z(p)$ on $S^{\ell-1}_{++}$ with image in $\R^{\ell}$. This function satisfies Walras' Law $z(p)\cdot p=0$ and can therefore be identified with a tangent vector field on $S^{\ell-1}_{++}$. Such vector fields can be identified with maps taking values in $\R^{\ell-1}$ by projecting each tangent space to the sphere orthogonally onto the plane $\{p\;:\;p_{1}+\dots+p_{\ell}=0\}\cong\R^{\ell-1}$.
 
When working with the open price simplex, as pointed out by Balasko~\cite{Balasko86}, there are constraints on which functions can be AEDs. 
Next we establish further necessary conditions for a function on the open simplex to be an AED.

We wish to consider the behaviour of (excess) demand functions as the $i$th price $p_i$ tends to zero from the right and the remaining ones are fixed and we write $p_i\to 0^+$ to indicate such limits. 
Since (excess) demand is homogeneous of degree zero in prices, when we normalize prices this corresponds to considering the limit along the curve on $S^{\ell-1}$ traced out by the corresponding ray through $p$ as $p_i$ tends to zero from the right.

\begin{proposition}
\label{prop:AED-BC}
Let $z:S_{++}^{\ell-1}\to\R^{\ell-1}$ be the AED of an economy. Then $z(p) = (z_1(p),\dots,z_\ell(p))$ satisfies the following boundary conditions as the $i$th price $p_i\to 0^+$:
\begin{enumerate}
    \item[(a)] For each $j=1,\dots,\ell$ with $j\neq i$, 
    $z_j(p)$ remains bounded as $p_i\to 0^+$. 
    \item[(b)] For each $i=1,\dots,\ell$,
    $\lim_{p_i\to 0^+}z_i(p) = +\infty$.
\end{enumerate}
\end{proposition}

\begin{proof}
Consider a consumer with utility function $u(x)$, fixed endowment $\omega$ and individual demand function $x(p)=x(p,\omega)$. 

We claim that for $j\neq i$ this consumer's demand for good $j$ remains bounded as $p_i\to 0^+$.
Indeed, if this were not the case there would be a sequence of prices of the form $p_n = (p_1,\dots,p_{i,n},\dots,p_\ell)$ with $p_{i,n}\to 0$ as $n\to\infty$ and $p_j$ independent of $n$ for $j\neq i$, and such that $x_j(p_n)>n$ for all $n$. 
But then we would have $p_jx_j(p_n)>p_n\cdot\omega$ for $n$ sufficiently large, since the latter remains bounded as $n\to\infty$ (in fact it converges to $\sum_{j\neq i}p_j\omega_j$). 
This means that the bundle $x(p_n)=(x_{1}(p_{n}),\dots,x_{\ell}(p_{n}))$ is not affordable at price $p_n$, which is a contradiction. 

It follows that this consumer's excess demand for good $j$ is also bounded as $p_i\to 0^+$, proving statement (a) of the proposition.

\medskip

For statement (b) assume that the consumer's endowment is not exclusively made up of good $i$. Such a consumer exists in the economy because of our assumption that all goods are present. This implies that the consumer's wealth remains bounded away from zero as the price of good $i$ goes to zero. Therefore, as $p_{i}\to 0^+$, the budget hyperplanes converge to a hyperplane parallel to, but distinct from, the hyperplane $p_{i}=0$. Now the boundary condition on utility implies that $x_{i}(p)\to\infty$ as $p_{i}\to 0^+$. Hence $\lim_{p_i\to 0^+}z_i(p) = +\infty$ as well.

\end{proof}

The preceding proposition implies that if we have a sequence of prices $p_{n}$ converging to the boundary of $S^{\ell-1}_{++}$ then $\norm{z(p_{n})}\to\infty$. Hence, by \cite[Propositon~2.2]{CDG2012}, we have the following corollary.

\begin{corollary}
\label{cor:proper}
Any AED $z\colon S^{\ell-1}_{++}\to\R^{\ell-1}$ is a proper map.
\end{corollary}

\begin{remark}
\label{rem:balasko}
It has been argued by Balasko~\cite{Balasko86} that there are natural topological constraints on AEDs defined on the open simplex. 
From our point of view this can be understood as follows. 
Since $S^{\ell-1}_{++}$ is homeomorphic to $\R^{\ell-1}$, the Brouwer degree of the proper map $z$ is a well defined integer. 
Viewing an AED $z$ as a tangent vector field on $S^{\ell-1}_{++}$, according to Proposition~\ref{prop:AED-BC}, this vector field is inward pointing along the boundary. If the singularities of $z$ are isolated, the Brouwer degree can be identified with the index of this vector field and it can be shown using the Poincar\'e--Hopf Theorem that this degree is $(-1)^{\ell-1}$ which agrees with Balasko's result.
\end{remark}

Let $C^k_{\pr}(\mathring{\Delta},\R^{\ell-1})$ denote the space of proper maps of class $C^k$ from $\mathring{\Delta}$ to $\R^{\ell-1}$.
In view of Corollary~\ref{cor:proper}, we define the set of AEDs with domain the open price simplex as
$$
\mathcal{Z}^*
  =\{z\in C^k_{\pr}(\mathring{\Delta},\R^{\ell-1}) \;:\; \text{$z$ is the AED of an economy}\}.
$$
Thus, in standard fashion, $\mathcal{Z}^*$ represents the set of economies.

The results of \cite{CDG2012} hold in the same way for the compact-open and Whitney topology because they are obtained for functions whose domain is a compact set, namely the trimmed price simplex $\Delta_\epsilon = \{p\in \Delta \;:\; p_i\geq \epsilon \}$ for some $\epsilon > 0$. 
Here, because we work on the open price simplex, we need to distinguish between the two topologies and we recall them next.
Let $C^k(N,P)$ be the set of maps $f:\;N\rightarrow P$ of class $C^k$ for $1\leq k \leq \infty$. 
In this space there are two natural topologies: the $C^m$ compact-open and the $C^m$ Whitney topology for $1\leq m \leq \infty$, where $m\leq k$. The Whitney topology is finer than the compact-open in the sense that it has more open sets.  Our definitions follow du Plessis and Wall \cite{dPW}. We start with $m=0$ and note that the $C^0$ topologies compare values of functions. The $C^m$ topologies compare values of derivatives up to order $m$ and therefore contain the topologies for all smaller degrees of differentiability. The basic open sets for the $C^0$ compact-open topology on $C^k(N,P)$ are subsets defined by
$$
A(K,U) = \{f | \; f(K) \subset U\},
$$
where $K \subset N$ is compact and $U \subset P$ is open. These open sets are used to define the basic open sets for the $C^0$ Whitney topology which are intersections of a collection of sets $A(K,U)$, $\cap_\alpha A(K_\alpha,U_\alpha)$ where $\{K_\alpha\}$ is a locally finite collection of compact subsets of $N$ and $\{U_\alpha\}$ is a collection of open subsets of $P$. The $C^m$ topologies, for $m > 0$, are defined by using not just the values of $f$ but also of its derivatives up to order $m$.

The concept of finite singularity type (FST) is presented in various forms in Section 2.4 of \cite{dPW} to which we refer the interested reader. For our results it suffices to think that a map is FST if it is locally equivalent to its Taylor polynomial up to a finite order. By locally equivalent we mean that there are coordinate changes in the source and the target spaces, defined in a small neighbourhood of a point and its value by $f$, that transform $f$ into its Taylor polynomial of finite order. This is a subtle concept in singularity theory and was used by the authors in \cite{CDG2012} where more detail can be found.

\section{Going to the boundary}
\label{sec:SMD-open-simplex}

For the genericity results of Section~\ref{sec:finiteness} below it is important to have additional information about the space of AED functions $\mathcal{Z}^*\subseteq C^k_{\pr}(\mathring{\Delta},\R^{\ell-1})$. 
As we have seen in Proposition~\ref{prop:AED-BC} and Remark~\ref{rem:balasko}, AEDs defined on the open price simplex satisfy additional constraints, apart from being proper.
Thus, the Sonnenschein-Mantel-Debreu result does not generalise in the obvious way to functions $z(p)$ defined on the open price simplex $\mathring{\Delta}$. 
Recall that up until now these results are shown to hold only on compact subsets of the open simplex. 
In this section we partially extend them to functions on $\mathring{\Delta}$. 

First we show that a positive multiple of an individual excess demand function is again the individual excess demand of a continuous consumer. 
The argument follows the idea outlined by Mas-Colell \cite[p.~193]{Mas-Colell85} and uses an integrability result relying on revealed preference theory. 

This allows us to use a finite number of Cobb-Douglas consumers to span an open subspace of AEDs defined all the way to the boundary. Openness allows us to prove genericity results in Section~\ref{sec:finiteness}.

For the arguments we present in this section it is convenient to normalise prices to belong to $S^{\ell-1}_{++}$, the positive part of the sphere since this is the natural mathematical setting for the proofs.

For the following result we need the version of the Strong Axiom of Revealed Preferences (SARP) for individual excess demand functions $z(p)$ found in Mas-Collell \cite[p.~192]{Mas-Colell85}. 
An excess demand function $z(p)$ satisfies SARP for excess demand if and only if $x(p,\omega) = z(p)+\omega$ satisfies SARP for demand (as defined in \cite{Mas-ColellWhinstonGreen}).

\begin{theorem}
\label{lem:multiple-of-ED}
Let $z(p)$ be an individual excess demand function obtained from utility maximisation. Let $\mu(p)$ be a real function for which there exist real constants $m,M$ such that $0<m<\mu(p)<M$ for all $p\in S^{\ell-1}_{++}$. Then $\mu(p)z(p)$ is the individual excess demand of a continuous consumer.
\end{theorem}

\begin{proof}
Since $z(p)$ is obtained from utility maximisation it is bounded below and satisfies SARP. Then, see Mas-Colell \cite[p.\ 193]{Mas-Colell85}, $\mu(p)z(p)$ satisfies SARP as well. 
Define $x(p,\omega)=\mu(p)z(p)+\omega$ for any endowment $\omega$. 
Since $\mu(p)z(p)$ satisfies SARP then so does $x(p,\omega)$.
By Houthakker \cite{Houthakker1950}, $x(p,\omega)$ is a  demand arising from a utility maximization problem -- see also Mas-Colell {\em et al.} \cite[Proposition 3.J.1]{Mas-ColellWhinstonGreen} and Aloqeili \cite{Aloqeili2005}; the latter paper provides the version of Houthakker's differential equation in terms of endowments, rather than wealth, as appropriate for our setup.
The bounds on $\mu(p)$ ensure that $x(p,\omega)$ satisfies the required boundary conditions as $p_{i}\to 0^+$, because $z(p)$ does. Hence the  utility function obtained by integration satisfies our assumptions and thus $\mu(p)z(p)$ is an individual excess demand.
\end{proof}

The previous result is attributed to Debreu by Chiappori and Ekeland \cite{ChiapporiEkeland2004} whose methods, however, require working on compact sets.

\begin{corollary}
\label{cor:mu-AED}
Let there be $\ell$ utility maximising consumers with individual excess demand functions $z_i:\;S^{\ell-1}_{++}\to\R^{\ell-1}$.
For any $\ell$ strictly positive functions $\mu_i\in C^k(S^{\ell-1}_{++},\R)$ satisfying the bounds of Theorem~\ref{lem:multiple-of-ED}, we have an AED defined by
\begin{displaymath}
  \sum_{i=1}^{\ell}\mu_i(p)z_i(p).
\end{displaymath}
\end{corollary}

\begin{theorem}
\label{thm:CD-spans}
There are $\ell$ Cobb--Douglas utility maximising consumers with individual excess demand functions $z_i^{CD}:\;S^{\ell-1}_{++}\to\R^{\ell-1}$ with the following property:
The set of functions of the form 
\begin{displaymath}
  z(p) = \sum_{i=1}^{\ell}\mu_i(p)z_i^{CD}(p),
\end{displaymath}
where $\mu_i\in C^k(S^{\ell-1}_{++},\R)$ is a collection of functions satisfying the bounds of Theorem~\ref{lem:multiple-of-ED}, is contained in $\mathcal{Z}^*$ and forms a neighbourhood of $z^{CD}(p) = \sum_{i=1}^{\ell}z_i^{CD}(p)$ in $C^k_{\pr}(S^{\ell-1}_{++},\R^{\ell-1})$ (in the Whitney topology).
\end{theorem}

\begin{proof}
Consider the Cobb--Douglas utility function of consumer $i$
\begin{displaymath}
  u_i(x_1,\dots,x_{\ell}) = x_1^{\alpha_{i1}}\dots x_\ell^{\alpha_{i\ell}}
\end{displaymath}
where the parameters $\alpha_{ij}$ are strictly positive and satisfy $\alpha_{i1}+\dots+\alpha_{i\ell} = 1$. The corresponding demand of consumer $i$ endowed with $\omega^i=(\omega^i_1,\dots,\omega^i_\ell)$ such that $\omega^i_j=0$ for $j \neq i$ is
\begin{displaymath}
    x_i(p,\omega^i) = \left(\alpha_{i1}\frac{p_i}{p_1}\omega^i_i,\dots, \alpha_{i\ell}\frac{p_i}{p_\ell}\omega^i_i \right)
\end{displaymath}
so that the individual excess demand of consumer $i$ is
\begin{align*}
  z_i^{CD}(p,\omega^i) 
  &=\left(\alpha_{i1}\frac{p_i}{p_1}\omega^i_i,\dots,-(1-\alpha_{ii})\omega^i_i,\dots,
  \alpha_{i\ell}\frac{p_i}{p_\ell}\omega^i_i\right).
\end{align*}
Note that all coordinates are strictly positive, except for the $i$th one, which is strictly negative. Thus, for each $i$, there is a hyperplane in the tangent space $T_pS^{\ell-1}_{++}$ such that the tangent vector $z_i^{CD}(p,\omega^i)$ points to one side of this hyperplane and for each $j\neq i$, the tangent vector $z_j^{CD}(p,\omega^j)$ points to the opposite side, as illustrated in the case $\ell=3$ in Figure~\ref{fig:tangent}.
\begin{figure}
  \centering
\begin{tikzpicture}[xscale=1, yscale=1]
%  \draw[gray,step=0.5] (-3.0,-2.0) grid (5.0, 3.0); % draw a grid to help determine the real bounding box
  \clip (-3.0,-2.0) rectangle (5.0, 3.0);
%\draw (-2.5, 2.7) node[right=-5pt] {$(a)$};
% axes
\draw[thick, ->] (0, 0) -- (4.5, 0);
\draw[thick, ->] (0, 0) -- (0,   3);
\draw[thick, ->] (0, 0) -- (-2.5,-1.4); % 
%names of axes
\draw (4.2, 0) node[above] {\small$p_2$};
\draw (-2.3, -1.3) node[below=2pt] {\small$p_1$};
\draw (0, 2.7) node[right] {\small$p_3$};
%\filldraw[black] (3,0) circle (2pt) node[below right] {\small$\xi_i$};
%\filldraw[black] (-1.75,-0.98) circle (2pt) node[above left=-2pt and -2pt] {\small$\xi_{i+1}$};
% sphere
\path (0, 2.0) edge[out=5,in=90,thick] (3.0,0.0);
\path (3.0, 0.0) edge[out=-120,in=-10,thick] (-1.8,-1.00);
\path (0.0, 2.0) edge[out=195,in=90,thick] (-1.8,-1.00);
% point
\filldraw[black] (0.5,1.0) circle (2pt);
% tangent plane clockwise
%horizontal top
\draw[thin] (0.5-1-0.5,1.0+0.8) -- (0.5+1,1.0+1.0+0.3);
%vertical right
\draw[thin] (0.5+1,1.0+1.0+0.3) -- (0.5+1+0.5,1.0-1+0.3);
%horizontal bottom
\draw[thin] (0.5+1+0.5,1.0-1+0.3) -- (0.5-1-0.2,1.0-1-0.3);
%vertical left
\draw[thin] (0.5-1-0.2,1.0-1-0.3) -- (0.5-1-0.5,1.0+0.8);
%vectors
\draw[thick, ->] (0.5, 1.0) -- (0.5+0.8, 1.0+0.8);
\draw[thick, ->] (0.5, 1.0) -- (0.5+0.8, 1.0-0.8);
\draw[thick, ->] (0.5, 1.0) -- (0.5-1.0, 1.0); % z_i^CD
\draw (0.5-1.0, 1.0) node[above] {\small$z_i^{CD}$};
%hyperplane
\draw[dashed] (0.5, 1.0) -- (0.5, 0);
\draw[dashed] (0.5, 1.0) -- (0.5, 2);
\end{tikzpicture}
  \caption{The tangent space is generated by positive linear combinations of three vectors. The dashed line corresponds to the hyperplane separating $z_i^{CD}(p)$ from the remaining vectors.}
  \label{fig:tangent}
\end{figure}
Hence any vector in the $(\ell-1)$-dimensional tangent space to $S^{\ell-1}_{++}$ can be written as a linear combination of these with strictly positive coefficients.
These coefficients can be taken to be bounded away from zero.

Now consider an open neighbourhood $\mathcal{V}$ of $z^{CD}\in C^k_{\pr}(S^{\ell-1}_{++},\R^{\ell-1})$ given by some fixed uniform bound on the values of the functions.
Then given any $z\in \mathcal{V}$ one can find $\ell$ functions $\mu_i\in C^k(S^{\ell-1}_{++},\R)$ satisfying the bounds of Theorem~\ref{lem:multiple-of-ED} such that
\begin{displaymath}
  z(p) = \sum_{i=1}^{\ell}\mu_i(p)z_i^{CD}(p).
\end{displaymath}
The bounds on the functions $\mu_i$ can be obtained in a neighbourhood of each $p$ and globalised by using a standard partition of unity argument.
It follows from Corollary~\ref{cor:mu-AED} that $z(p)$ is an AED.
\end{proof}

We conjecture that the construction of Theorem~\ref{thm:CD-spans} holds for utility functions that are not Cobb-Douglas, so long as the endowments are chosen in the same way. 

\section{Finiteness of equilibria}
\label{sec:finiteness}

In this section we extend the results of \cite{CDG2012} to the Whitney topology which allows us to make a stronger genericity statement for finiteness of equilibria, even for critical economies.

\begin{theorem}
\label{thm:open-dense-FST}
The subspace of smooth FST maps is open and dense in $C^k_{\pr}(\mathring{\Delta},\R^{\ell-1})$.
\end{theorem}  

\begin{proof}
The proof follows an analogous strategy to that of Theorem~5.1 in
\cite{CDG2012}.
The key point is the following: the open simplex $\mathring{\Delta}$ has a locally finite covering by closed balls. 

In view of the definition of the Whitney topology this means that the proof of Proposition~3.3 in \cite{CDG2012} immediately adapts to show that the subspace $T_{W,\mathrm{pr}}$ of that proposition is actually \emph{open} in the Whitney topology. It follows, as in Theorem~5.1 and its  Corollary~5.3 in \cite{CDG2012}, that FST maps are open and dense in the Whitney topology.
\end{proof}

\begin{theorem}
\label{thm:AED-FST}
Any subset of $\mathcal{Z}^*$, open in $C^k_{\pr}(\mathring{\Delta},\R^{\ell-1})$, has an open and dense set of AEDs with only a finite number of equilibria.
\end{theorem}

\begin{proof}
Call the given open set $\mathcal{Z}'$ and write $C^k_{\prfst}(\mathring{\Delta},\R^{\ell-1})$ for the subspace of smooth FST maps in $C^k_{\pr}(\mathring{\Delta},\R^{\ell-1})$.
Then it follows from Theorem~\ref{thm:open-dense-FST} that the intersection of $C^k_{\prfst}(\mathring{\Delta},\R^{\ell-1})$ with $\mathcal{Z}'$  is open and dense in $\mathcal{Z}^*$ for the Whitney topology.
Since any $f$ of FST has a finite number of zeros the result follows.
\end{proof}

Let $\mathcal{Z}^*_{\mathrm{crit}}\subset\mathcal{Z}^*$ denote the
subspace of critical economies. 
These are economies for which the Jacobian of the AED is singular. The next theorem guarantees that even when an economy is critical, the FST prevents the system from exhibiting a infinite continuum of equilibria\footnote{We thank an anonymous referee for this comment.}.

\begin{theorem}
\label{thm:AED-FST-crit}
Any subset of $\mathcal{Z}^*_{\mathrm{crit}}$, open in $C^k_{\pr}(\mathring{\Delta},\R^{\ell-1})$, has an open and dense set of AEDs with only a finite number of equilibria.
\end{theorem}

\begin{proof}
It follows from Theorem 5.4 of \cite{CDG2012} that the subset of $\mathcal{Z}^*_{\mathrm{crit}}$ which are of FST is dense in $\mathcal{Z}^*_{\mathrm{crit}}$. 
Since $\mathcal{Z}^*_{\mathrm{crit}}\subset\mathcal{Z}^*$ has the subspace topology, openness follows from Theorem~\ref{thm:open-dense-FST}. 
The rest of the proof is analogous to that of the previous theorem.
\end{proof}

The following result shows that finiteness of equilibria remains a generic property for a class of economies defined on the open simplex.

\begin{theorem}
The interior of the neighbourhood of AEDs on $\mathring{\Delta}$ given in Theorem~\ref{thm:CD-spans} has an open and dense subset whose economies have a finite  number of equilibria.
\end{theorem}

\begin{proof}
Since the interior is open this follows from Theorem~\ref{thm:AED-FST}.
\end{proof}

Finally, in the following theorem we show that the hypothesis of Theorem~\ref{thm:AED-FST} can be satisfied by adding $\ell$ consumers to the economy. 

\begin{theorem}
Any economy with a finite number $\ell'$ of consumers is contained in another economy with $\ell+\ell'$ consumers, itself contained in an open set of economies.
\end{theorem}

\begin{proof}
Identifying the economies with their AEDs we can add the $\ell$ consumers of Theorem~\ref{thm:CD-spans} to obtain an economy with the stated property.
\end{proof}

We remark that, if $\mathcal{Z}^*$ were known to be an open subset of $C^k_{\pr}(\mathring{\Delta},\R^{\ell-1})$, genericity of finiteness of equilibria would follow for both regular and critical economies by Theorems~\ref{thm:AED-FST} and \ref{thm:AED-FST-crit}.

\section{Final remarks}

The results presented above partially answer, and open new perspectives on, a question that has persisted for over half a century. They also provide sound foundations for decision-making in economics based on the possibility of distinguishing between equilibrium states that are distinct, and in finite number.

Although the interest in general equilibrium theory is not at its peak, several aspects of existence and finiteness of Walrasian equilibria continue to receive attention from the scientific community. Recently the existence question has been treated with relaxed assumptions by Podczeck
and Yannelis \cite{PodczeckYannelis2022}, Khan {\em et al.}\ \cite{KhanMcLeanUyanik2025}, and Anderson and Duanmu \cite{AndersonDuanmu2025}. The computation of equilibria has been addressed by Gauthier {\em et al.}\ \cite{GauthierKehoeQuintin2022} who provide an algorithm to construct examples with multiple Walrasian equilibria. Also, Cheung {\em et al.}\ \cite{CheungColeDevanur2020} determine when a gradient descent method can be used to compute the equilibria.

The existence of equilibria has also been discussed by Anderson and Duanmu \cite{AndersonDuanmu2025b} in the context of quota and emission tax equilibria. 
The presence of multiple equilibria gives rise to the inequivalence of these two approaches to regulate carbon emissions.

Our results here and in \cite{CDG2012} provide a solid foundation for the continued quest for methods to calculate and study of Walrasian equilibria, as well as exploring the question of their finiteness for economies defined on the open price simplex.

\paragraph{Acknowledgements:}
We thank Y. Balasko and the participants in the Economic Theory and Applications Reading Group of FEP.UP, namely J. Correia-da-Silva and A. Rubinchik.
We also thank S. Dakhlia for introducing us to this problem.

The current version benefitted from a pertinent question of an anonymous referee identifying a gap in our argument. 

The authors were partially supported by CMUP, member of LASI, which is financed by national funds through FCT -- Funda\c{c}\~ao para a Ci\^encia e a Tecnologia, I.P., under the project UID/00144/2025 -- \url{https://doi.org/10.54499/UID/00144/2025}.
Much progress on this work was made while both authors were visiting Myggedalen, Greenland, whose hospitality is gratefully acknowledged.

\end{document}